\documentclass[english]{iopart}
\usepackage[T1]{fontenc}
\usepackage[latin9]{inputenc}
\usepackage{geometry}
\usepackage{amsthm}
\usepackage{amstext}
\usepackage{amssymb}
\usepackage{graphicx}
\usepackage{esint}

\makeatletter

\DeclareFontEncoding{LGR}{}{}
\DeclareTextSymbol{\~}{LGR}{126}

\usepackage{iopams}
\usepackage{setstack}
\theoremstyle{plain}
\newtheorem{thm}{\protect\theoremname}
  \theoremstyle{plain}
  \newtheorem{conjecture}[thm]{\protect\conjecturename}
  \theoremstyle{plain}
  \newtheorem{lem}[thm]{\protect\lemmaname}

\usepackage[numbers, sort&compress]{natbib}

\newcommand{\eqref}[1]{(\ref{#1})}

\makeatother

\usepackage{babel}
  \providecommand{\conjecturename}{Conjecture}
  \providecommand{\lemmaname}{Lemma}
\providecommand{\theoremname}{Theorem}

\begin{document}

\title[On integrability of chipping models ...]{On integrability of zero-range chipping models with factorized steady
state}

\author{A.M. Povolotsky}

\address{{\large Bogoliubov Laboratory of Theoretical Physics, Joint Institute
for Nuclear Research, 141980 Dubna, Russia}}

\address{{\large National Research University Higher School of Economics,
20 Myasnitskaya Ulitsa, Moscow 101000, Russia}}

\ead{alexander.povolotsky@gmail.com}
\begin{abstract}
Conditions of integrability of general zero range chipping models
with factorized steady state, which were proposed in {[}Evans, Majumdar,  Zia
2004 J. Phys. A \textbf{37} L275{]},
are examined. We find a three-parametric family of hopping probabilities
for the models solvable by the Bethe ansatz, which includes most of
known integrable stochastic particle models as limiting cases. The
solution is based on the quantum binomial formula for two elements
of an associative algebra obeying generic homogeneous quadratic relations,
which is proved as a byproduct. We use the Bethe ansatz to solve an
eigenproblem for the transition matrix of the Markov process. On its
basis we conjecture an integral formula for the Green function of evolution operator for the model on
an infinite lattice and derive the Bethe equations for the spectrum
of the model on a ring.
\end{abstract}



%\ams{44}

\submitto{\JPA}

\maketitle

\section{Introduction}

Integrability is a key feature of stochastic particle systems, which
allows one to obtain plenty exact results. Often in a proper (scaling)
limit these results become meaningful in a context of a whole universality
class. The most prominent example is the asymmetric simple exclusion
process (ASEP) \cite{Ligget}. Its exact Bethe ansatz solution yielded
the dynamical exponent for Kardar-Parisi-Zhang (KPZ) universality
class \cite{Gwa Spohn}, the crossover function for the transition
from KPZ to Edwards-Wilkinson (EW) regime \cite{Kim}, the universal
large deviation function for a particle current \cite{Derrida Lebowitz},
e.t.c..The integrability of the totally asymmetric simple exclusion
process (TASEP) was a starting point for calculation of the universal
correlation functions in infinite systems belonging to KPZ class \cite{Johansson,Rakos Schuetz,Sasamoto Nagao,Sasamoto}.
Finally the full solution for time dependent evolution in the partially
asymmetric exclusion process (PASEP) \cite{Tracy Widom 1} and subsequent
calculation of the distribution of tagged particle position \cite{Tracy Widom 2}
culminated in the exact solution of the KPZ equation \cite{Sasamoto Spohn,Amir Corwin Quastel}. The same
result was also obtained from studies of polymer in random media \cite{Calabrese Le Doussal Rosso},
 based on the solution \cite{Dotsenko Klumov,Dotsenko} of  another integrable model, Lieb-Liniger bosons with delta interaction \cite{Lieb Liniger} (for review of the whole story see \cite{Corwin}).

The integrability imposes restrictive limitations on dynamical rules
governing an evolution of stochastic particle systems. Though the
concept of universality extends the range of applicability of the
results obtained for existing models, a limited choice of such models
makes search for new integrable dynamics an important challenging
problem. It is of interest to include new interactions, that could
be used to check stability of universal quantities against modifications
of the dynamical rules. Several integrable models generalizing the
ASEP were proposed. There are ASEP-like models with long range jumps
\cite{Alimohammad Karimipou  Khorrami,Alimohammad Karimipou  Khorrami 2},
an interacting diffusion without exclusion \cite{sasamoto wadati,sasamoto-wadati 2},
the models with non-local avalanche dynamics \cite{Priezzhev  Ivashkevich Povolotsky Hu},
zero range processes \cite{Povolotsky}, e.t.c. For discrete time
dynamics there were several versions of updates proposed: random sequential
\cite{Lee Kim}, backward sequential \cite{Priezzhev Pramana}, parallel \cite{Povolotsky Mendes},
sub-lattice parallel \cite{Priezzhev Poghosyan Schuetz} and generalized
\cite{Derbyshev Poghosyan  Povolotsky  Priezzhev} update, each having
its own history and appeared in different contexts and for various
applications \cite{Rajewsky  et al}. For every mentioned integrable
model the Markov matrix of transition probabilities governing the
time evolution can be diagonalized by the Bethe ansatz, which makes
calculation of many physical quantities of interest possible, at least
in principle.

Stochastic systems of interacting particles on the lattice with product
stationary measure attracted significant attention \cite{Evans}.
The reason is that when the stationary measure has a simple form of
the product of one-site factors, the observables over generally non-equilibrium
stationary states can be evaluated using the toolbox of equilibrium
statistical mechanics of non-interacting particle systems. The simplest
example is continuous time ASEP where the stationary measure is a
product of one-site Bernoulli measures. More complex case is zero
range process (ZRP) where a single particle can jump to the neighboring
site with probability depending only on occupation number of the site
of departure. The stationary state of this model was shown to be rich
enough. In particular it demonstrates a real space condensation transition
for certain choice of hopping probabilities \cite{Evans Haney}. The
most general dynamics with on-site interaction leading to a factorized
stationary state was considered by Evans et. al. \cite{evans majumdar zia}.
They found necessary and sufficient condition for existence of the
product stationary measure in a class of models with multiparticle
chipping dynamics. These conditions prescribe a certain functional
form for hopping probabilities.

The question we address in the present paper is: What is the most
general integrable version of the latter model? Similar question was
addressed earlier with respect to a particular case of this model, zero range process, first
with continuous \cite{Povolotsky} and later with discrete time \cite{Povolotsky Mendes}
dynamics. As a result the processes were obtained with hopping probabilities
depending on two parameters and expressed in terms of so-called q-numbers.
The result of the present paper is a three parametric family of hopping
probabilities having a the functional form proposed in \cite{evans majumdar zia},
which ensure the integrability of the Markov dynamics. The hopping
probabilities are obtained from the requirement that the Markov matrix
is diagonalizable by the coordinate Bethe ansatz. We obtain eigenvectors
and eigenvalues of the Markov matrix. For an infinite lattice they can be used to construct
the Green function, the transition probabilities between particle configurations for arbitrary
time, provided that  the generalized completeness relation for eigenvectors is proved. We state a
conjecture for this relation, and derive a Green function out of it.
  In the case of periodic   boundary conditions the eigenvectors are expressed in terms of solutions of the
Bethe equations we derive. We also show that the dynamics obtained includes
all the models mentioned above as particular limiting cases.

The article is organized as follows. In section \ref{sec:Model-and-results}
we formulate the model and state our main result, an expression of
the hopping probabilities. In subsection \ref{sub:Limiting-cases.}
we give an extensive survey of the models appeared in the literature
before, which can be obtained as limiting cases of our model. This
subsection is not related to the rest of the article. The Reader not
acquainted with the history of the subject may skip over subsection \ref{sub:Limiting-cases.}
on first reading. In section \ref{sec:Transfer-matrix} we describe
the method of construction of the Markov matrix, such that the transition
probabilities of the form proposed in \cite{evans majumdar zia}
providing an existence of the factorized steady state
are  restricted to ensure the integrability
of the model. The latter condition suggests that all the interactions
are introduced via two-particle boundary conditions. The procedure
of the reduction of many-particle interactions to the two-particle
ones can be restated as the problem of writing the binomial formula
for two elements of an associative algebra obeying generic uniform
quadratic relations. Theorem \ref{theorem: quant binom}, stated in subsection \ref{sub:Generalized-quantum-binomial}
is the result of the solution of this problem. The proof of Theorem
\ref{theorem: quant binom} is carried over to \ref{sec: app1}. Section
\ref{sec:Bethe-ansatz} is devoted to application of the Bethe ansatz
 to diagonalization of the Markov matrix constructed. As a consequence,
we state a conjecture for an integral formula  of the Green function
of an evolution operator for an infinite lattice (subsection \ref{sub:Infinite-lattice-and})
 and derive the Bethe equations for the spectrum of the evolution operator
of the system on a ring (subsection \ref{sub:Discrete-spectrum-on}).
In subsection \ref{sub:ZRP-ASEP-transformation} we generalize the
Bethe ansatz to a model with the exclusion interaction related to our
model. Some concluding comments are given in section \ref{sec:Conclusion}.

\section{Model and results\label{sec:Model-and-results}}

Consider particles on a one-dimensional lattice. They live in sites
of the lattice with no restrictions on the number of particles at
a site. A configuration of particles is uniquely specified by a collection
of occupation numbers $\mathbf{n}\equiv\{n_{i}\}_{i\in\mathcal{L}}$,
where $n_{i}\in\mathbb{Z}_{\geq0}$ and the set $\mathcal{L}$ is either $\mathbb{Z}$
for an infinite lattice or $\mathbb{Z\mathrm{\mathit{/L\mathbb{Z}}}}$
for a ring of size $L$. In the latter case we impose periodic boundary
conditions $n_{1}\equiv n_{L+1}$. The system evolves in discrete
time according to the following dynamical rules. At each time step
$m$=$0,\ldots,n$ particles from a site occupied with $n$ particles
jump to the next site on the right with probability $\varphi(m|n)$,
which satisfies the normalization condition
\begin{equation}
\sum_{m=0}^{n}\varphi(m|n)=1.\label{eq:stochastic}
\end{equation}
We suggest that the update is parallel, i.e. at every time step all
sites are updated simultaneously. Given initial probability distribution
of particle configurations $P_{0}(\mathbf{n})$, the system is characterized
by the probability $P_{t}(\mathbf{n})$ for the system to be in configuration
$\mathbf{n}$ at subsequent moments of time $t$. This probability
obeys master equation
\[
P_{t+1}(\mathbf{n})=\sum_{\mathbf{n'}}\mathbf{M_{\mathbf{n,}n'}}P_{t}(\mathbf{n'}),
\]
with transition matrix \textbf{M} defined by the above dynamical rules:
\begin{equation}
\mathbf{M_{n,n'}}=\sum_{\{m_{k}\in\mathbb{Z}_{\geq0}\}_{k\in\mathcal{L}}}\prod_{i\in\mathcal{L}}T_{n_{i},n_{i}'}^{m_{i-1},m_{i}},\label{eq:Markov matrix}
\end{equation}
where
\[
T_{n_{i},n_{i}'}^{m_{i-1},m_{i}}=\delta_{(n_{i}-n_{i}'),(m_{_{i-1}}-m_{i})}\varphi(m_{i}|n_{i}')
\]
 and we define $\phi(m|n)=0$ for $m>n$. The dynamics conserves the
total number of particles, that is to say that the matrix $\mathbf{M}$
is block-diagonal with blocks indexed by the number of particles on
the lattice. Within the blocks corresponding to any finite number
of particles the transition probabilities are well defined. In the
following we will work within the sector with this number fixed and
finite, $\sum_{i\in\mathcal{L}}n_{i}=N<\infty$. The stationary state
is the right eigenvector of the matrix $\mathbf{M}$ corresponding
to the largest eigenvalue $\Lambda=1.$ Its existence is ensured by
the stochasticity condition (\ref{eq:stochastic}), which is equivalent
to the fact that there is a corresponding left eigenvector with all
components equal to one. The stationary state is unique, provided
that the fixed particle number blocks of $\mathbf{M}$ are non-degenerate.
In the latter case for a finite lattice the state vector can be normalized
to have a meaning of stationary probability measure, with each vector
component giving the probability of corresponding configuration. That
the components are real is ensured by the Perron-Frobenius theorem.
The stationary state is the state the system eventually arrives at
in the large time limit. On the infinite lattice, due to translation invariance, there are no
stationary probability measures that exhibit a finite number of
particles in a typical configurations. However, the stationary measure
(unnormalized) still can be defined by components of stationary state eigenvector.
It is going to play an important role in the subsequent analysis.

It was shown in \cite{evans majumdar zia} that the stationary measure
$P_{st}(\mathbf{n})$ is a product measure
\begin{equation}
P_{st}(\mathbf{n})=\prod_{i\in\mathcal{L}}f(n_{i}),\label{eq:stationary state}
\end{equation}
if and only if there exist two functions $w(m)$ and $v(m)$, such
that
\begin{equation}
\varphi(m|n)=\frac{v(m)w(n-m)}{\sum_{k=0}^{n}v(k)w(n-k)}.\label{eq: phi(n|m)}
\end{equation}
Then the one-site weights $f(n)$ will read as
\begin{equation}
f(n)=\sum_{k=0}^{n}v(k)w(n-k).\label{eq:normalization}
\end{equation}
 Hence the question we address is the following. What should be the
form of the functions $w(m)$ and $v(m)$ for the matrix $\mathbf{M}$
to define an integrable model? Below we use the Bethe ansatz to diagonalize
the matrix $\mathbf{M}$. Its applicability imposes certain constraints
on the form of $v(k)$ and $w(k)$. In a nutshell the procedure consists
in a solution of one- and two-particle problems, which is always possible
for arbitrary $\varphi(m|n)$, while those for three and more particles
must be in a sense reduced to the former ones. Therefore, assuming
that the one- and two-particle hopping probabilities $\varphi(1|1),\varphi(2|1)$
and $\varphi(2|2)$ can take arbitrary values, we uniquely fix the
three parameter family of jumping probabilities $\varphi(m|n)$, which
is the main result of the present paper. With more convenient parametrization
in terms of three real numbers $q,\mu$ and $\nu$ we obtain the following
expressions for functions $v(n)$ and $w(n)$

\begin{equation}
v(k)=\mu^{k}\frac{(\nu/\mu;q)_{k}}{(q;q)_{k}},\,\,\,\, w(k)=\frac{(\mu;q)_{k}}{(q;q)_{k}},\label{eq:v(k), w(k)}
\end{equation}
where the notation $(a;q)_{n}$ is used for q-Pochhammer symbol,

\[
(a;q)_{n}=\left\{ \begin{array}{ll}
\prod_{k=0}^{n-1}(1-aq^{k}), & n>0\text{;}\\
1, & n=0;\\
\prod_{k=1}^{|n|}(1-aq^{-k})^{-1},\,\, n<0.
\end{array}\right.
\]
As discussed in \ref{sec: app1}, the one-site weight following from
(\ref{eq:normalization}) is

\begin{equation}
f(n)=\frac{(\nu;q)_{n}}{(q,q)_{n}}\label{eq:f(n)}
\end{equation}
and the jumping probabilities are

\begin{equation}
\varphi(m|n)=\mu^{m}\frac{(\nu/\mu;q)_{m}(\mu;q)_{n-m}}{(\nu;q)_{n}}\frac{(q;q)_{n}}{(q;q)_{m}(q;q)_{n-m}}.\label{eq:distrib}
\end{equation}
The following range of the parameters $q,\mu$ and $\nu$ is such that $\varphi(m|n)$
is a probability distribution in $m$. It is enough that the functions $v(m)$ and
$w(m)$ are nonnegative reals for all $m$. In particular this always
holds when $0\leq \nu\leq\mu$ and $|q|<1$.  However, in some cases the
model remains meaningful beyond this range. Below we will assume this range everywhere,
if the opposite is not stated explicitly.

The function given in (\ref{eq:distrib}) has been known as the weight
function associated with q-Hahn polynomials. Otherwise, to our knowledge,
it was not used, neither as hopping probability in interacting particle
models, nor in a more general probabilistic context. However, many
of its limiting cases were.

For further discussion we introduce the notations for a few other
q-analogues. These are q-number
\[
[n]=\frac{1-q^{n}}{1-q},
\]
q-factorial
\[
[n]!=[1]\times\dots\times[n],\,\,\,[0]!=1
\]
and q-binomial coefficient
\[
{n \brack m}=\frac{[n]!}{[m]![n-m]!}=\frac{(q;q)_{n}}{(q;q)_{m}(q;q)_{n-m}}.
\]
They turn into the usual number, factorial and binomial coefficient
respectively in the limit $q\to1,$ while the q-Pochhammer symbol
is related to the ordinary Pochhammer symbol $(a)_{n}$ by limiting
transition $(q^{a},q)_{n}/(1-q)^{n}\to(a)_{n}$.

One can see that the second fraction in (\ref{eq:distrib}) is a q-binomial
coefficient and the jumping probabilities obtained are reminiscent
of those defining the binomial distribution. Indeed, in the limit $q\to1$ they converge
to the regular binomial probabilities
\[
\lim_{q\to1}\varphi(m|n)=p^{m}(1-p)^{n-m}\left(\begin{array}{c}
n\\
m
\end{array}\right),
\]
which are the probabilities of $m$ successes in a series of  $n$ independent Bernoulli trials
with the success probability
\begin{equation}
p=\frac{\mu-\nu}{1-\nu}\label{eq:p}.
\end{equation}
Thus, the original
formula (\ref{eq:distrib}) is a two-parametric deformation of the
binomial distribution. Several deformations of the binomial distribution
were discussed in the probabilistic literature \cite{Johnson Kemp Kotz}.
Some of them can be obtained from our hopping probabilities by limiting
transitions. Specifically, the two-parametric distribution obtained
from $\varphi(m|n)$ in the limit $\mu\to 0$ and $\nu<0$ gives a distribution
of the number of successes in $n$ independent Bernoulli trials, where the success
probability depends on the number of the trial, so that the odds of success (probability of success divided by  probability of failure) geometrically decreases, $\theta_i\equiv p_i/(1-p_i)=pq^{i-1}$. It was considered in \cite{Kemp Kemp}
as a candidate for stochastic model for the dice throwing data, and
also can be obtained as a stationary distribution for a random process
describing dynamics of dichotomized parasite populations \cite{Newton Kemp}.
The distribution obtained in the limit $\nu\to 0$ was proposed in
\cite{Jing Fan,Jing} in order to construct a q-binomial state, which
interpolated between the coherent and particle number states of q-oscillator.
Both these limiting expressions of $\varphi(m|n)$ can be reinterpreted
as probability for $m$ particles to be absorbed, when $n$ particles
cross a field with random number of absorption points (traps), given
the distributions of the number of absorption points are q-analogues
of the Poisson distribution, Heine and Euler distributions respectively
\cite{Charalambides}.

\subsection{Limiting cases.\label{sub:Limiting-cases.}}

Let us consider how the known stochastic particle models are related
to our model. To this end, we first note that there are two dualities
connecting models with seemingly different dynamics. The first one,
that we refer to as ZRP-ASEP transformation, relates a system where
the number of particles at a site is unbounded (ZRP-like) to a system
where either zero or one particle at a site is allowed (ASEP-like).
The transformation consists in replacing a site with $n$ particles
by a string (compact cluster) of $n$ sites, occupied by one particle
each, plus one empty site ahead, see Fig.~\ref{fig:zrp-asep}(a).
\begin{figure}[h]
\centerline{\includegraphics[width=0.4\textwidth]{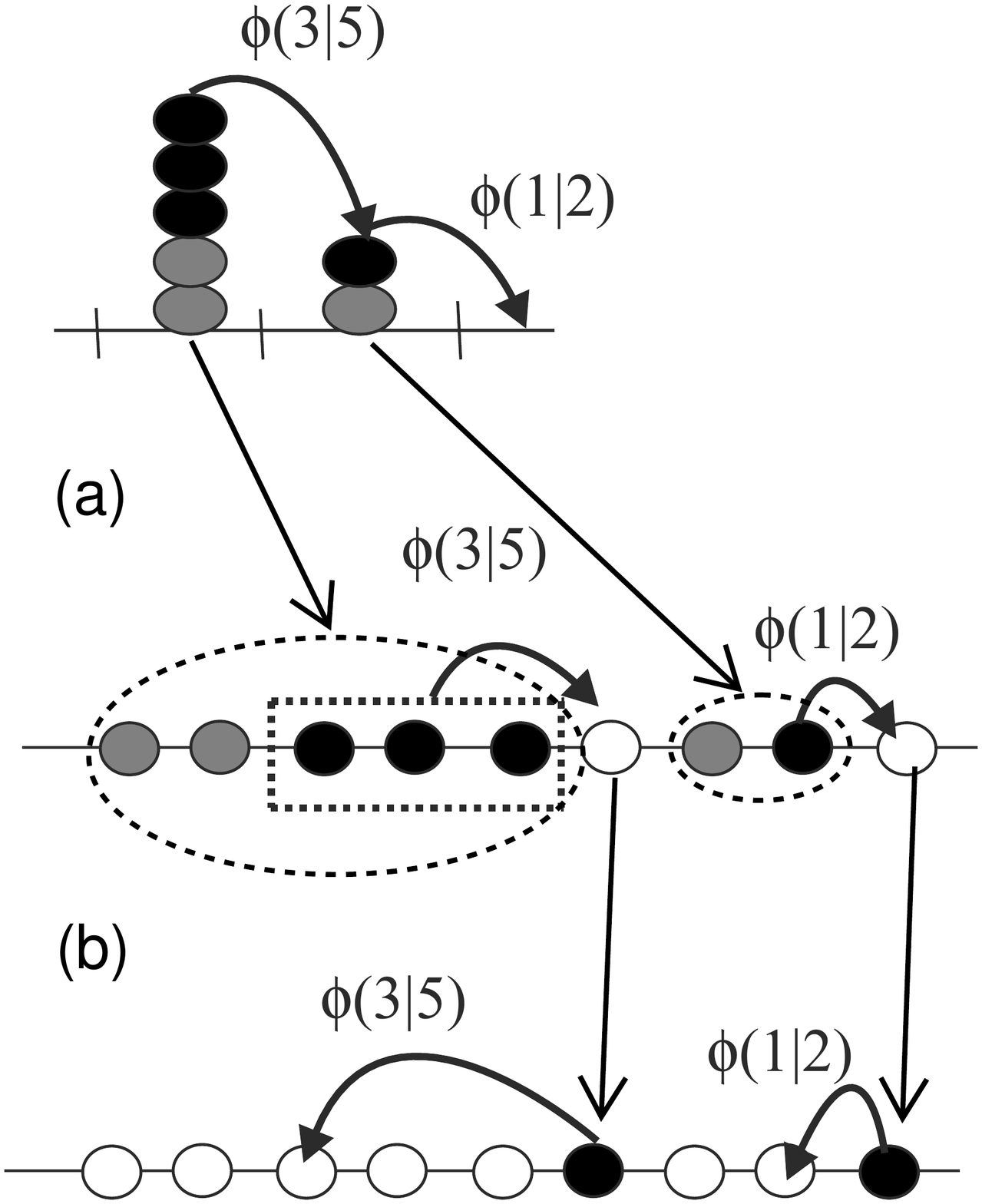}}
\caption{ZRP-ASEP (a) and particle-hole (b) transformations \label{fig:zrp-asep}}
\end{figure}
 Correspondingly, $m$ particles jumping from a site with $n$ particles
to the next site are replaced by a one-step shift of an $m-$particle
cluster detached from an $n-$particle cluster. The number of sites
of the lattice in the ASEP-like system is equal to the number of sites
in the ZRP-like system plus the number of particles.

The second duality is the particle-hole transformation, which relates
two ASEP-like systems, Fig.~\ref{fig:zrp-asep}(b). It interchanges the occupied and empty sites.
A jump of a particle corresponds to a shift of the cluster of holes
in the opposite direction. With this comment in mind we sketch a list
of known models that can be obtained as limiting cases of our hopping
probabilities (\ref{eq:distrib}).

\subsubsection*{Non-interacting particles, $q=1.$ }

It was noted above that the $q\to1$ limit of the hopping probabilities
gives us the binomial distribution. This corresponds to the free non-interacting
particles, each performing the Bernoulli random walk independently
of the others. That is to say that every time step all particles attempt
to make jumps with the same probability $p$ given in (\ref{eq:p}).
The binomial coefficient counts the number of ways to choose $m$
out of $n$ particles in a site. Note that taking the limit $q\to1$
reduces the dependence on the two other parameters to a single parameter
$p$. When $q\neq1,$ the parameters $\nu$ and $\mu$ are responsible
for an interaction between particles, which can be either repulsive
or attracting, accelerating or deceleration the global motion. The
Bernoulli random walks evolving with discrete time $t$ can be transformed
into continuous time Poissonian random walks by taking limits $p\to0$,
i.e. $\mu\to\nu$, and $t\to\infty$ simultaneously, so that new continuous
time $\tau=tp$ remains finite.

\subsubsection*{TASEP with generalized update, $q=0.$ }

In the ZRP picture particles jump from a site occupied by $n$ particles
with probabilities $\varphi(0|n)=(1-p)$, $\varphi(m|n)=(1-\mu)p\mu^{m-1}$ for
$0<m<n$ and $\varphi(n|n)=p\mu^{n-1}$. After ZRP-ASEP transformation
this limit reproduces the process proposed in \cite{Derbyshev Poghosyan  Povolotsky  Priezzhev}.
Consider a version of the TASEP, with backward sequential update,
where a particle jumping to a site remembers whether this site was
empty or occupied before the current time step. Specifically, during
an update each cluster of particles is scanned from the rightmost
to the leftmost particle. The first particle makes an attempt to jump
forward with probability $p.$ In the case of success the second one
tries to jump with probability $\mu$, generally different from $p$,
and so do the third, forth e.t.c.. If eventually a particle has failed
to jump, all the subsequent particles within the same cluster will
stay with probability one. If $\mu=p,$ we obtain the usual TASEP
with backward sequential update, while $\mu=0$ corresponds to the
parallel update case. Within the range $0<\mu<1$ the effective interaction
varies from repulsive to attractive, with the limit $\mu\to1$ corresponding
to particles sticking together. The limit $p\to0$ is the continuous
time version of the TASEP if $\mu\sim p$ and the continuous time
fragmentation model when $\mu$ stays finite.

\subsubsection*{Multiparticle hopping asymmetric diffusion and long range hopping
models, $\nu\to\mu=q$.}

If in this limit, if we set $p=dt$, the hopping probabilities simplify
to $\varphi(m|n)\simeq dt/[n]_{1/q}$, where the subscript $1/q$
indicates that the q-number deformation parameter is $q^{-1}$, rather
than $q$. The model with similar hopping probabilities was first
proposed in \cite{sasamoto-wadati 2}. It, however, admitted particle
jumps in both directions, and the asymmetry strength was rigidly
bound to parameter $q$. Its generalization, where the asymmetry strength
is unrelated to $n-$dependence of the hopping probabilities, was
considered in \cite{lee}. Its totally asymmetric version is given
by the limit under consideration. After ZRP-ASEP transformation we
obtain the ASEP-like model, proposed in \cite{Alimohammad Karimipou  Khorrami},
where a particle pushes its right neighbouring particles to the right
with the rate $r_{n}=1/[n]_{1/q}$, depending on the number $n-1$
of these particles. The model interpolates between the continuous
time TASEP and the drop-push model \cite{Schuetz Ramaswamy Barma},
corresponding to the limits $q\to\infty$ and $q\to0$ respectively.

\subsubsection*{q-bosonic ZRP, q-TASEP and asymmetric avalanche process, $\mu=q\nu$.}

In this case the first q-Pochhammer symbol $(\nu/\mu;q)_{m}$ vanishes
as soon as $m>1$. Therefore, only single particle jumps remain allowed.
The process obtained is a discrete time ZRP, where one particle jumps
from a site occupied by $n$ particles with probability $\varphi(1|n)=p\cdot[n]$.
The corresponding integrable model, q-boson model, was first discovered
in \cite{Bogoliubov Bullough} in the language of the algebraic Bethe
ansatz. Its Hamiltonian (continuous time) version was later discussed
as an interacting particle model in \cite{sasamoto-wadati 2}. It
appeared again in \cite{Povolotsky}, where the question addressed
was: What are the most general hopping probabilities, which make the
totally asymmetric continuous time ZRP model integrable? Later the
discrete time model was obtained by addressing the same question to
the discrete time ZRP \cite{Povolotsky Mendes}. Again, the continuous
time model can be obtained from the discrete time one by taking the
limit $p\to0.$ For the discrete time model the parameters take their
values in the range $|q|<1,0<p<1,$ while in the continuous time case
$q$ can be any real number. In the latter case the limit $q\to\infty$
corresponds to the drop-push model, where a particle goes to the next
vacant site on the right jumping over all its neighbors after exponentially
distributed waiting time or, equivalently, pushes all its right neighbors
one step to the right. After the ZRP-ASEP and particle-hole transformations
the q-bosonic ZRP becomes so called q-TASEP, where $\varphi(1|n)$
is a probability for a particle to jump one step forward, given its
headway  is $n$. The q-TASEP appeared recently as a limiting case
of the Macdonald process \cite{Borodin Corwin} and was used to study
a semi-discrete polymer in a random media \cite{Borodin Corwin Ferrari}.

Another interesting continuous time limit of this model, the Asymmetric
Avalanche Process (AAP) \cite{Priezzhev  Ivashkevich Povolotsky Hu},
is obtained in the limit $\nu\to0$ and $|q|<1$. Setting $1-p=dt$
and going to a moving reference frame, which shifts one step to the
right every time step, we obtain an ASEP-like model, where the transitions
between particle configurations are described in terms of non-local
avalanche dynamics. Specifically, in the moving frame the continuous
time dynamics looks as follows. Starting from a configuration with
at most one particle at every site any particle can jump to the left
neighbouring site after exponentially distributed waiting time. If
a particle meets another particle at the target site it can carry
the latter along with itself to the next site on the left with probability
$(1-[2])$ or leave it and go further alone. In general, if in course
of the avalanche $n>1$ particles are found at the same site, either
all $n$ particles go to the next site on the left with probability
$(1-[n])$ or otherwise one particle stays and $(n-1$) particles
go. Thus, at every step the number of particles in the avalanche can
either decrease or increase by one or stay unchanged. The avalanche
ends when one particle from a pair goes to an empty site. The discrete
time avalanche dynamics is considered instant in the slow Poissonian
time scale, so that the avalanches plays the role of transitions between
ASEP-like particle configurations. The interest to this model was
caused by the phase transition from  intermittent to continuous
flow, which takes place in the infinite system at critical value of
the density of particles $\rho_{c}=1/(1-q).$ The $q=0$ limit of
AAP is again the drop-push model, in which, however, the particles
jump to the opposite direction with respect to the one mentioned above.

\subsubsection*{Geometric q-TASEP, $\nu=0.$}

Very recently a preprint \cite{Borodin Corwin-1} by Borodin and Corwin
appeared, where two version of discrete time TASEP were proposed.
 One of them is the so-called geometric q-TASEP, where
a particle is allowed to jump forward to any vacant site between it
and the next particle. The probability $p_{n,\mu}(l)$ of the jump
length $l$ depending on the headway $n$ can be obtained from our
$\varphi(l|n)$ by setting $\nu=0$. The process can be obtained from
our model by making ZRP-ASEP and particle hole transformations. Note
that the processes discussed in \cite{Borodin Corwin-1} were obtained
as a reduction of the Macdonald process, and the technique was developed
for a particular case of evolution with step initial condition. A
question was also posed whether the Bethe ansatz technique is available
to study the same problem, which would allow a consideration of other
initial and boundary conditions. The present paper answers this question
giving even more general form of hopping probabilities.

\vphantom{}Other limiting cases of our model can be considered, for
which the hopping probabilities simplify. We mentioned those that
appeared in the literature before. In addition, several models with
partially asymmetric dynamics were proposed like the PASEP \cite{Gwa Spohn},
two-parametric long range hopping model \cite{Alimohammad Karimipou  Khorrami 2},
multiparticle diffusion without exclusion \cite{sasamoto-wadati 2},
the Push-ASEP \cite{Borodin Ferrari} and the AAP with two-sided hopping
\cite{Povolotsky Priezzhev Hu}. They can not be directly obtained
as limiting cases of our totally asymmetric model. However, generality
of the model makes us expect that being interpreted as a transfer
matrix of a quantum integrable model our Markov matrix can generate
also the Hamiltonians describing jumps in both directions, see e.g.
\cite{Van Diejen}.

\section{Transfer matrix\label{sec:Transfer-matrix}}

We are going to find the conditions for the eigenproblem of the Markov
matrix
\[
\Lambda\mathbf{\Psi}=\mathbf{M}\mathbf{\Psi}
\]
to be solvable by the Bethe ansatz. Here $\mathbf{M}$ is the matrix
defined in (\ref{eq:Markov matrix},\ref{eq: phi(n|m)}), \textbf{$\mathbf{\Psi}$}
--- a column vector, and $\Lambda$ is an eigenvalue. Our solution
is based on the following observation. It is natural to expect that
the groundstate, i.e. the eigenstate of the transfer matrix corresponding
to the largest eigenvalue $\Lambda_{0}=1$, is the state of highest
symmetry and, in particular, is translationally invariant. This is
indeed the case for many models solved before. As the Bethe ansatz,
which is supposed to give the eigenvectors, is an oscillating function,
the groundstate Bethe vector should be zero momentum eigenstate, i.e.
constant for all particle configurations. On the other hand, the groundstate
of the chipping model with hopping probabilities of the form (\ref{eq: phi(n|m)})
is the product stationary state (\ref{eq:stationary state}). However,
the left eigenvector corresponding to the groundstate has exactly
the required form $\bar{\mathbf{\Psi}}_{0}^{\mathbf{T}}=(1,\dots,1)$,
where the superscript \textbf{$\mathbf{T}$ }refers to the matrix
transposition transforming a column into a row. Therefore, we may
try the Bethe ansatz to find the solution of the left eigenproblem.
\[
\Lambda\bar{\mathbf{\Psi}}^{\mathbf{T}}=\bar{\mathbf{\Psi}}^{\mathbf{T}}\mathbf{M}
\]
or equivalently
\[
\Lambda\bar{\mathbf{\Psi}}=\mathbf{M^{T}}\bar{\mathbf{\Psi}}.
\]
The key observation, first made in \cite{Povolotsky Priezzhev}, is
that the matrix $\mathbf{M}$ of the form (\ref{eq:Markov matrix},\ref{eq: phi(n|m)})
is related to its transpose $\mathbf{M^{T}}$ by simple conjugation
\begin{equation}
\mathbf{M^{T}=\Pi SMS^{-1}\Pi},\label{symmetry}
\end{equation}
where $\mathbf{S}$ is the diagonal matrix with elements
\[
S_{\mathbf{n,n'}}=\delta_{\mathbf{n,n'}}/P_{st}(\mathbf{n})
\]
and $\mathbf{\Pi}$ is the parity transformation reversing the order
of sites or equivalently the direction of motion. Indeed, consider
matrix element $\mathbf{M_{n,n'}}$ corresponding to the transition
from a configuration $\mathbf{n}$ to a configuration $\mathbf{n'}$.
In fact, on a subset with fixed number of particles
the sum in (\ref{eq:Markov matrix}) consists of the only term
\[
\mathbf{M_{n,n'}}=P_{st}(\mathbf{n})^{-1}\prod_{i\in\mathcal{L}}v(m{}_{i})w(n_{i}-m_{i}),
\]
where $m_{i}\geq0$ is the number of particles jumping from site $i$
to site $(i+1)$. Once $\mathbf{n}$ and $\mathbf{n'}$are given,
the numbers $m_{i}$ can be uniquely determined from the system of
equations
\begin{equation}
n_{i}-m_{i}=n'_{i}-m_{i-1},\, i\in\mathcal{L}.\label{eq:n_i-m_i}
\end{equation}
Only those matrix elements are nonzero, which yield non-negative $m_{i}$
for all $i\in\mathcal{L}$. Conjugation with matrix $\mathbf{S}$
affects the matrix elements of \textbf{$\mathbf{M}$} by replacing
the factor $P_{st}(\mathbf{n})^{-1}$ by $P_{st}(\mathbf{n}')^{-1}$:
\begin{equation}
(\mathbf{SMS}^{-1})_{\mathbf{n,n'}}=P_{st}(\mathbf{n'})^{-1}\prod_{i\in\mathcal{L}}v(m{}_{i})w(n_{i}-m_{i}).\label{eq:SMS}
\end{equation}
On the other hand, the matrix elements of \textbf{$\mathbf{M^{T}}$},
which can be thought of as transition weights of the time reversed
process, are
\begin{equation}
(\mathbf{M^{T}})_{\mathbf{n,n'}}=\mathbf{M_{n',n}}=P_{st}(\mathbf{n'})^{-1}\prod_{i\in\mathcal{L}}v(m'{}_{i})w(n'_{i}-m'_{i}),\label{eq:M^T}
\end{equation}
 where $m'_{i}=m_{i-1}$ is the number of particles one must transfer
back to site $(i-1)$ from site $i$ to return from $\mathbf{n'}$
to $\mathbf{n}$. Taking into account (\ref{eq:n_i-m_i}) and the
translation symmetry of the lattice, we see that the weights (\ref{eq:SMS})
and (\ref{eq:M^T}) exactly coincide. The only difference is that
in the time reversed process the particles move  in the opposite direction.

Thus, we are going to solve the eigenproblem for matrix $\mathbf{M^{0}\equiv\mathbf{SMS}}^{-1}$,
defined by matrix elements (\ref{eq:SMS}). Once its eigenvectors
$\mathbf{\Psi^{0}}$ have been found, the right eigenvectors of $\mathbf{M}$
are $\mathbf{\mathbf{\Psi}}=\mathbf{S}^{-1}\mathbf{\Psi^{0}}$ and the
left eigenvectors can be obtained by applying the parity transformation
$\bar{\mathbf{\Psi}}=\mathbf{\Pi\Psi^{0}}$. Specifically, we are
looking for functions $v(k)$ and $w(k)$ that ensure the Bethe ansatz
solvability of the eigenproblem for $\mathbf{M^{0}}$.

Before going into calculations, we note that the hopping probabilities
$\varphi(m|n)$ are invariant with respect to simultaneous transformations
$v(k)\to a\theta^{k}v(k),\,\, w(k)\to b\theta^{k}w(k)$, where $a,b$
and $\theta$ are arbitrary nonzero constants. This three-parametric
freedom can be removed by imposing three constraints on these functions.
For example we can fix the functions $w(k)$ and $v(k)$ at three
values of arguments (two for one of them and one for the other). Now
we choose two of them as
\begin{equation}
v(0)=w(0)=1.\label{eq:v(0),w(0,v(1)}
\end{equation}
Thus, we fix the stationary weight of empty site, $f(0)=1$. Before
fixing the third constraint we note that $\varphi(m|n)$ can be represented
as a function of ratios $v(n)/(v(1))^{n}$ and $w(n)/(v(1))^{n}$
rather than on $v(n)$ and $w(n)$ alone. Therefore fixing the value
of $v(1)$ is equivalent to fixing an exponential part of functions
$w(k)$ and $v(k).$ This will be done in \ref{sec: app1}, when we
go to a more convenient parametrization.

We also should note that in general the stationary measure $P_{st}(\,)$
constructed as a product (\ref{eq:stationary state}) is not normalized
even in the finite system. If we want a probability measure, the overall
normalization factor, called the partition function, must be evaluated.

As usual in the coordinate Bethe ansatz technique, we first consider
the one-particle problem. Then, the two-particle problem looks as
a direct product of one-particle problems in the range of particle
coordinates, where the interaction is absent. The interaction, which
reveals itself only at the border of physical domain of particle coordinates,
can be accounted for as boundary conditions. Then, one has to consider
many-particle problem with three and more particles on the lattice.
The condition of the Bethe ansatz solvability is that all the many
particle interactions are introduced via the two-particle boundary
conditions.

\subsubsection*{One particle. }

A representation of particle configurations equivalent to the one
used above can be given in terms of positions of particles on the
lattice. From now on we specify an $N-$particle configuration by
a set of weakly increasing coordinates of particles
\begin{equation}
\mathbf{x}=(x_{1},\leq\dots,\leq x_{N}).\label{eq:coord range}
\end{equation}
For one particle on the lattice the whole configuration is a single
particle coordinate $x_{1}\equiv x$. Then the eigenproblem reads
as follows
\begin{equation}
\Lambda_{1}\Psi^{0}(x)=p\Psi^{0}(x-1)+(1-p)\Psi^{0}(x),\label{eq:free}
\end{equation}
where $p\equiv\varphi(1|1)=v(1)/(v(1)+w(1)).$ The corresponding stationary
weights are $f(1)=v(1)+w(1).$ As we discussed above, the parameter
$p$ depends only on the ratio $w(1)/v(1)$.

\subsubsection*{Two particles.}

Now we have to consider the cases with two particles located at different
sites, $x_{1}<x_{2}$, and at the same site, $x_{1}=x_{2}\equiv x$,
separately. Inspecting the expression (\ref{eq:SMS}) of the matrix
elements of $\mathbf{M^{0}},$ we find out that in the first case
they depend on parameters of the dynamics via $p$ (i.e. via $w(1)/v(1)$)
and, in fact, look like the two independent one-particle problems
\begin{eqnarray}
\Lambda_{2}\Psi^{0}(x_{1,}x_{2}) & = & (1-p)[p\Psi^{0}(x_{1}-1,x_{2})+(1-p)\Psi^{0}(x_{1,}x_{2})],\label{eq:two patrticles free}\\
\, & + & p[p\Psi^{0}(x_{1}-1,x_{2}-1)+(1-p)\Psi^{0}(x_{1,}x_{2}-1)].\nonumber
\end{eqnarray}
If this form was valid in the whole range of particle coordinates,
there would not be any more complications comparing to the one-particle
equation (\ref{eq:free}). In the case $x_{1}=x_{2}\equiv x$, however,
the non-interacting form breaks up, and the new parameters $w(2)$
and $v(2)$ (in fact $w(2)/(v(1))^{2}$ and $v(2)/(v(1)){}^{2}$)
come into the game
\begin{equation}
\Lambda_{2}\Psi^{0}(x,x)=f(2)^{-1}[w(2)\Psi^{0}(x,x)+v(1)w(1)\Psi^{0}(x-1,x)+v(2)\Psi^{0}(x-1,x-1)],\label{eq:tw particles interact}
\end{equation}
where $f(2)=w(2)+v(1)w(1)+v(2).$ To restore the free equation (\ref{eq:two patrticles free})
let us formally rewrite it for the case $x_{1}=x_{2}\equiv x$. We
notice that term $\Psi^{0}(x,x-1)$ appears, which is beyond the physical
domain (\ref{eq:coord range}). In the following we refer to terms
of this kind as forbidden and to those within the physical domain
as allowed. It is our choice to assign the value to the forbidden
term in such a way, that it compensates the difference between free
equation (\ref{eq:two patrticles free}) and interacting one (\ref{eq:tw particles interact}).
\begin{equation}
\Psi^{0}(x,x-1)=\alpha\Psi^{0}(x-1,x-1)+\beta\Psi^{0}(x-1,x)+\gamma\Psi^{0}(x,x),\label{eq:boundary conds}
\end{equation}
 where
\begin{equation}
\alpha=\frac{v(2)/f(2)-p^{2}}{p(1-p)},\,\beta=\frac{v(1)w(1)/f(2)}{p(1-p)}-1,\,\gamma=\frac{w(2)/f(2)-(1-p)^{2}}{p(1-p)}.\label{eq:a,b,c}
\end{equation}
 The equation (\ref{eq:two patrticles free}) supplied with the boundary
conditions (\ref{eq:boundary conds}) completely define the two-particle
problem.

\subsubsection*{$N$ particles.}

For arbitrary number of particles the equations should in principle
include all the parameters $w(n)$ and $v(n)$ for $n=1,\dots,N,$
which generally can be arbitrary. The integrability, however, restricts
the choice. To make the problem solvable we try to represent our equations
as the equations for non-interacting particles with suitable boundary
conditions in the same vein as we did for the two-particle case. The
basic condition of the Bethe ansatz solvability is all the boundary
conditions being of the same form (\ref{eq:boundary conds}). This
fact reduces the set of independent parameters to those three we have
already used.

An example of the procedure for three particles, $N=3$, is as follows.
First, when we write down the equations for three particles with $\Lambda_{3}\Psi^{0}(x_{1},x_{2},x_{3})$
in the l.h.s., we find three different cases to be considered: three
particles in different sites, i.e. $x_{1}<x_{2}<x_{3}$, one particle
in one site and two in another, $x_{1}=x_{2}<x_{3}$ or $x_{1}<x_{2}=x_{3},$
and all the three particles in the same site, $x_{1}=x_{2}=x_{3}\equiv x$.
The first case, is already the equation for three independent particles.
The second one is a combination of one- and two-particle problems,
(\ref{eq:free}) and (\ref{eq:tw particles interact}), which can
be converted to the non-interacting form by applying the two-particle
boundary conditions (\ref{eq:boundary conds}) to the pairs of coordinates
in inverse order, e.g. $(x,x-1)$. An essentially new case is the
equation with $\Psi^{0}(x,x,x)$ in the l.h.s.. Again, we would like
to replace it by the equation for three independent particles. If
we write corresponding non-interacting equation, it will contain four
forbidden terms in the r.h.s: $\Psi^{0}(x,x,x-1)$, $\Psi^{0}(x,x-1,x-1)$,
$\Psi^{0}(x,x-1,x)$ and $\Psi^{0}(x-1,x,x-1).$ The idea is to express
them in terms of the allowed configurations only using the boundary
conditions of the form (\ref{eq:boundary conds}). Note that if we
simply apply our boundary conditions to the pairs of coordinates $(x,x-1)$,
some of the terms we obtain will be forbidden again. However, they
will be found among the four terms we have just mentioned. In fact,
the relations we will obtain in this way can be treated as the system
of four linear equations for four forbidden terms, which, having been
solved, yields the forbidden terms expressed via the allowed terms.
The solution must be substituted into the non-interacting equation,
so that only the allowed terms remain. Then, we compare the coefficients
coming with the allowed terms with the coefficients in the true interacting
equation and try to identify the values of $v(3)$ and $w(3)$, relying
on the expectation that the solution for the hopping probabilities
of the suggested form (\ref{eq: phi(n|m)}) exists.

The procedure for arbitrary $N$ is similar. We want to transform
the equation for non-interacting particles to the equation for interacting
particles using the two-particle boundary conditions. Remarkably,
the transformation we did to the transition matrix resulted in the
transition coefficients, which factorize into a product of single-site
terms, which depend only on the number of particles coming to a site
and on the number of particles in this site after the transition.
To illustrate this fact consider the transition in which site $x$
becomes occupied by $n$ particles after $k$ particles have arrived
from site $(x-1)$ and some particles may have jumped out. The $k$
particles that have jumped in bring the factor $\nu(k)$, while those
$(n-k)$ that have stayed bring the factor $w(n-k)$. The overall
denominator is $f(n)$ independently of the previous state of the
site. Therefore, the equation with $\Lambda\Psi^{0}(\dots,x^{n},\dots)$
in the l.h.s. will contain the sum
\begin{equation}
\sum_{k=0}^{n}\varphi(k|n)\Psi^{0}(\dots,(x-1)^{k},x{}^{n-k},\dots)\label{eq:interact n-particles}
\end{equation}
 on the right, where $x^{n}$ means a string of $n$ letters $x$,
i.e. $n$ particles in the site $x$ and the coefficients are supposed
to be of the form $\varphi(n|k)=v(k)w(n-k)/f(n)$, the same as in
(\ref{eq: phi(n|m)}). For several occupied sites we have products
of similar terms summed independently of each other. The corresponding
part of the non-interacting equation is
\begin{equation}
\sum_{k_{1}=0}^{1}\dots\sum_{k_{n}=0}^{1}p^{k_{1}+\dots+k_{n}}(1-p)^{n-(k_{1}+\dots+k_{n})}\Psi^{0}(\dots,x-k_{1},\dots,x-k_{n},\dots).\label{eq:non-interact n-particles}
\end{equation}
Our aim is to reduce one equation to the other by iterative application
of the two-particle boundary conditions. As a result we will get the
arguments of all terms $\Psi^{0}(\,)$ ordered so that all the symbols
$(x-1)$ appear on the left of the symbols $x$.

\subsection{Generalized quantum binomial\label{sub:Generalized-quantum-binomial}}

The problem can be formalized as that of the generalized quantum binomial.
Consider an associative algebra generated by two elements $A$ and
$B$, which obey a general homogeneous quadratic relation
\begin{equation}
BA=\alpha AA+\beta AB+\gamma BB.\label{eq: relations}
\end{equation}
Within the set of all words made of the symbols $A$ and $B$ we distinguish
a subset of normally ordered words, where all symbols $A$ are put
on the left of all symbols $B$, i.e. where no combination $BA$ is
present. An arbitrary homogeneous element of the algebra can be represented
as a linear combination of normally ordered words of the same degree,
obtained by repetitive application of the relation (\ref{eq: relations}).
That this representation is unique is guaranteed by the diamond lemma
\cite{diamond lemma}. A particularly interesting example of the normally
ordered representation is a non-commutative analogue of the Newton
binomial:
\begin{equation}
(A+B)^{n}=\sum_{0\leq k\leq n}\mathcal{C}_{k}^{n}A^{k}B^{n-k},\label{eq:quantum binomial}
\end{equation}
where $\mathcal{C}_{k}^{n}$ are the generalized binomial coefficients
depending on the parameters of the defining relation. In purely commutative
case, $\alpha=\gamma=0,\,\beta=1$, and in the case of $q-$commuting
variables , $\alpha=\gamma=0,\, b=q,$ $\mathcal{C}_{k}^{n}$ are
well known to be the usual binomial and the $q-$binomial coefficients,
respectively. We are interested in the case of generic coefficients
$\alpha,\beta,\gamma.$ Indeed, let us associate $A$ with $(x-1)$
and $B$ with $x$ in (\ref{eq:interact n-particles}, \ref{eq:non-interact n-particles}).
The boundary conditions (\ref{eq:boundary conds}) used to get rid
of the forbidden combinations $(\dots,x,x-1,\dots)$ act just like
the defining relations (\ref{eq: relations}). What we need is to
construct the following ``skew'' binomial sum

\begin{equation}
(pA+(1-p)B)^{n}=\sum_{m=0}^{n}\varphi(m|n)A^{m}B^{n-m},\label{eq:skew binomial}
\end{equation}
where coefficients $\varphi(n|k)$ are nothing but the jumping probabilities
to be defined. In principle, instead of defining the parameters $\alpha,\beta,\gamma$
in terms of one and two particle dynamics, we could go the other way
around, starting from assigning them any complex values considered
as input data. The resulting coefficients $\varphi(n|m)$ would define
the matrix $\mathbf{M}^{0},$ which is still diagonalizable by the
Bethe ansatz. Then, however, the problem might lose its probabilistic
content, though possibly could still be treated as some quantum or
statistical physics model. In our case the values of $\alpha,\beta,\gamma$
read from (\ref{eq:a,b,c}) satisfy relation $\alpha+\beta+\gamma=1,$
which remove one degree of freedom. On the other hand, the parameter
$p$ in the l.h.s. of (\ref{eq:skew binomial}) yields another degree
of freedom, so that we again have three free parameters, e.g. $v(2)$,
$w(2)$ and $p$ or $w(1).$ Note that (\ref{eq:skew binomial}) can
be reduced to (\ref{eq:quantum binomial}) by absorbing the parameter
$p$ into $A$ and/or $B$ and changing the defining relations accordingly,
which return us to the generic case. Also, the range of the values
of $\alpha,\beta,\gamma$ is limited by the condition that the coefficients
$\varphi(n|k)$ have the meaning of hopping probabilities, i.e. $v(n)\geq0$
and $w(n)\geq0$ for any $n\geq0.$ As we do not know whether the
generalized binomial formula for the case of generic homogeneous quadratic
relations appeared in the literature before, we state it as a theorem.
The expression of the generalized binomial coefficients (aka hopping
probabilities $\varphi(m|n)$ ) is a main result of the present paper.
The proof of this theorem is brought to \ref{sec: app1}.
\begin{thm}
\label{theorem: quant binom}Consider an associative algebra over
complex numbers with two generators $A$,$B$. Suppose the generators
satisfy the homogeneous quadratic relation (\ref{eq: relations}),
where $\alpha,\beta,\gamma$ are arbitrary complex numbers constrained
by $\alpha+\beta+\gamma=1.$ Then, for any complex number $p,$ the
coefficients \textup{$\varphi(m|n)$} of the binomial sum (\ref{eq:skew binomial})
are given by the formula (\ref{eq:distrib}), where $q,\nu$ and $\mu$
give a convenient parametrization for $\alpha,\beta,\gamma$ and $p$:

\begin{equation}
\alpha=\frac{\nu(1-q)}{1-q\nu},\,\beta=\frac{q-\nu}{1-q\nu},\,\gamma=\frac{1-q}{1-q\nu}\label{eq:a,b,c (parametrized)}
\end{equation}

and
\begin{equation}
\mu=p+\nu(1-p),\label{eq:mu}
\end{equation}
 and we suppose that $\nu\neq q^{-k}$, for any $k\in\mathbb{N}.$
\end{thm}
One can see that $\varphi(m|n)$ has indeed the product form (\ref{eq: phi(n|m)}).
For $n=2$ it complies with the definition (\ref{eq:a,b,c}).

\section{Bethe ansatz\label{sec:Bethe-ansatz}}

Now we are in a position to diagonalize the matrix $\mathbf{M}^{0}$.
The eigenproblem is reformulated as the free equation
\begin{equation}
\Lambda_{N}\Psi^{0}(\mathbf{x})=(1-p)^{n}\sum_{\mathbf{k}\in\{0,1\}^{\otimes N}}\left(\frac{p}{1-p}\right)^{||\mathbf{k}||}\Psi^{0}(\mathbf{x}-\mathbf{k}),\label{eq:master free full}
\end{equation}
where $||\mathbf{k}||=k_{1}+\dots+k_{N}$, supplied with the boundary
conditions
\begin{eqnarray}
\Psi^{0}(\dots,x,x-1,\dots) & = & \alpha\Psi^{0}(\dots,x-1,x-1,\dots)\nonumber \\
\, & + & \beta\Psi^{0}(\dots,x-1,x,\dots)+\gamma\Psi^{0}(\dots,x,x,\dots),\label{eq:boundary full}
\end{eqnarray}
where the parameters $\alpha,\beta,\gamma$ are given in (\ref{eq:a,b,c (parametrized)})
expressed in terms of $q$ and $\nu$. We are looking for an eigenfunction
in form of the Bethe ansatz
\begin{equation}
\Psi^{0}(\mathbf{x|z})=\sum_{\sigma\in S_{N}}A_{\sigma}\hat{\sigma}\mathbf{z}{}^{\mathbf{x}},\label{eq:Bethe ansatz}
\end{equation}
that depends on $N-$tuple of ``quantum'' numbers $\mathbf{z}=(z_{1},\dots,z_{N}).$
Here the summation is performed over the set $S_{N}$ of all permutations
$\sigma=(\sigma_{1},\dots,\sigma_{N})$ of the numbers $1,\dots,N$,
the hat symbol indicates the action of an element of the permutation
group on the functions of $N-$tuple $\mathbf{z}$, $\hat{\sigma}\mathbf{z}=(z_{\sigma_{1}},\dots,z_{\sigma_{N}})$
and $\hat{\sigma}\mathbf{z^{x}}=(z_{\sigma_{1}}^{x_{1}},\dots,z_{\sigma_{N}}^{x_{N}})$,
and $A_{\sigma}$ are the coefficients to be defined, indexed by permutations.
Substituting this ansatz into the equation (\ref{eq:master free full})
we obtain the eigenvalue as a function of the parameters $\mathbf{z}$,
\begin{equation}
\Lambda_{N}(\mathbf{z})=\prod_{i=1}^{N}\Lambda_{1}(z_{i}),\label{eq:eigenvalue}
\end{equation}
which is a product of one-particle eigenvalues
\begin{equation}
\Lambda_{1}(z)=1-p+p/z.\label{eq:Lambda_1}
\end{equation}
The boundary conditions yield the S-matrix, the ratio of two coefficients
$A_{\sigma}$ corresponding to permutations differing from each other
in an elementary transposition of two neighbors,
\begin{equation}
S(z_{i},z_{j})\equiv\frac{A_{\dots ij\dots}}{A_{\dots ji\dots}}=-\frac{\alpha+\beta z_{i}+\gamma z_{i}z_{j}-z_{j}}{\alpha+\beta z_{i}+\gamma z_{i}z_{j}-z_{j}}.\label{eq:S-matrix}
\end{equation}
Given the initial condition $A_{id}=1$ for the identical permutation
$id=(1,\dots,N)$, this can be solved to
\begin{equation}
A_{\sigma}=\mathrm{sgn}(\sigma)\prod_{1\leq i<j\leq N}\frac{\alpha+\beta z_{\sigma_{i}}+\gamma z_{\sigma_{i}}z_{\sigma_{j}}-z_{\sigma_{j}}}{\alpha+\beta z_{i}+\gamma z_{i}z_{j}-z_{j}},\label{eq:A_sigma}
\end{equation}
where $\mathrm{sgn}(\sigma)$ is the permutation sign.

To write the above formulas in a shorter form we make a variable change
\begin{equation}
z_{i}=\frac{1-\nu u_{i}}{1-u_{i}}.\label{eq:z-u change}
\end{equation}
 Then, the $S$-matrix simplifies to
\begin{equation}
S(u,v)=\frac{qv-u}{v-qu},\label{eq:S(u,v)}
\end{equation}
 the form familiar from studies in quantum integrable systems, and
the one-particle eigenvalue in new variables looks as follows
\begin{equation}
\Lambda_{1}(u)=\frac{1-\mu u}{1-\nu u}.\label{eq:Lambda_1(u,v)}
\end{equation}
Hence we have
\begin{equation}
\Lambda_{N}=\prod_{i=1}^{N}\left(\frac{1-\mu u_{i}}{1-\nu u_{i}}\right),\label{eq:Lambda(u)}
\end{equation}
and the components of the eigenvector of $\mathbf{M}^{0}$are
\begin{equation}
\Psi^{0}(\mathbf{x}|\mathbf{z})=\sum_{\sigma\in S_{N}}\mathrm{sgn}(\sigma)\prod_{i=1}^{N}\prod_{\begin{array}{c}
j>i\end{array}}\frac{u_{\sigma_{i}}-qu_{\sigma_{j}}}{u_{i}-qu_{j}}\frac{\left(1-\nu u_{\sigma_{i}}\right)^{x_{i}}}{\left(1-u_{\sigma_{i}}\right)^{x_{i}}}.\label{eq:Psi^0(u)}
\end{equation}
This is used to write the right and left eigenvectors of $\mathbf{M}$,
which, according to the discussion in the beginning of the section,
are obtained by maltiplying $\Psi^{0}(\mathbf{x}|\mathbf{z})$ by
$P_{st}(\mathbf{x})$ and by applying parity transformation to the
spacial coordinates respectively:
\begin{equation}
\Psi(\mathbf{x}|\mathbf{z})\sim P_{st}(\mathbf{x})\Psi^{0}(\mathbf{x}|\mathbf{z}),\,\,\,\bar{\Psi}(\mathbf{x|}\mathbf{z})\sim\Psi^{0}(\mathbf{\hat{\Pi}x}|\mathbf{z}).\label{eq:Psi(u)}
\end{equation}
Here, the result of the action of parity transformation  applied
to $\Psi^{0}(\mathbf{x}|\mathbf{z})$ is replacement of particle
coordinates $x_{i}$ to $-x_{i}$ and, correspondingly, inverting
the order of particles $i\to N-i+1$, i.e. $\hat{\Pi}\mathbf{x}=(-x_{N},\dots,-x_{1})$.
Note that the proportionality sign ``$\sim$'' reflects the fact that
the components of eigenvectors are defined up to an arbitrary $\mathbf{z}$-dependent
factor, which can be fixed by normalization conditions.

The spectrum of parameters $\mathbf{z}$ depends on the type of the
lattice. The infinite lattice and the ring should be considered separately.

\subsection{Infinite lattice and Green function conjecture \label{sub:Infinite-lattice-and}}

On the infinite lattice the parameters $\mathbf{z}$ can take any
values. In practice, what we want is to use the eigenfunctions to
expand the solutions of the master equation with given initial conditions.
 Specifically, the eigenvectors of the matrix $\mathbf{M}$ can be
used as an analogue of the Fourier basis. Given the probability distribution
$P_{t}(\mathbf{x})$ we would like to represent it as an integral
\begin{equation}
P_{t}(\mathbf{x})=\int\tilde{P}_{t}(\mathbf{z})\Psi(\mathbf{x}|\mathbf{z})\mathcal{M}(d\mathbf{z}),\label{eq:Fourier}
\end{equation}
where the measure  $\mathcal{M}(\,)$ and the domain of integration
have to be chosen consistent with initial conditions. Given a function
$\tilde{P}_{0}(\mathbf{z})$ that provides the integral representation
at time $t=0$, the time dependence of the Fourier coefficients directly
follows from the fact that $\Psi(\mathbf{x}|\mathbf{z})$ is an eigenvector
of the matrix $\mathbf{M}$:
\[
\tilde{P}_{t}(\mathbf{z})=\left[\Lambda_{N}(\mathbf{z})\right]^{t}\tilde{P}_{0}(\mathbf{z}).
\]
The choice of the measure and the domain of integration is verified
by examining a particular case of the initial conditions, $P_{0}(\mathbf{x})=\delta_{\mathbf{x,x^{0}}}$,
while the other initial distributions can be considered as linear
combinations of delta functions. In this case the Fourier coefficient
$\tilde{P}_{0}(\mathbf{z})$ is expected to be proportional to $\bar{\Psi}(\mathbf{y|}\mathbf{z})$,
the component of the left eigenvector of the matrix $\mathbf{M}$,
while the relation (\ref{eq:Fourier}) at $t=0$ follows from the
generalized completeness relation:

\begin{equation}
\int\bar{\Psi}(\mathbf{y|}\mathbf{z})\Psi(\mathbf{x}|\mathbf{z})\mathcal{M}(d\mathbf{z})=C_{N}\delta_{\mathbf{x},\mathbf{y}},\label{eq:completeness}
\end{equation}
where $C_{N}$ is a normalization constant. Using (\ref{eq:Psi(u)})
we write the relation in the following form
\begin{equation}
\int\Psi^{0}(\mathbf{\hat{\Pi}y}|\hat{R}\mathbf{z})\Psi^{0}(\mathbf{x}|\mathbf{z})\mathcal{M}(d\mathbf{z})=\frac{C_{N}\delta_{\mathbf{x},\mathbf{y}}}{P_{st}(\mathbf{x})},\label{eq:completeness-1}
\end{equation}
where $\hat{\Pi}\mathbf{y}=(-y_{N},\dots,-y_{1})$, and $R\mathbf{z}=(z_{N},\dots,z_{1})$.\footnote{Here, for further convenience we use an inversion $\hat{R}\mathbf{z}$
of the $N$-tuple $\mathbf{z}$. In fact, the effect of the action
of any permutation $\hat{\sigma}$ applied to $\Psi^{0}(\mathbf{x}|\mathbf{z})$
(acting on components of $\mathbf{z}$) is a multiplication of this
function by a function of $\mathbf{z}$ but not of $\mathbf{x}$,
$\hat{\sigma}\Psi^{0}(\mathbf{x}|\mathbf{z})=A_{\sigma}^{-1}\Psi^{0}(\mathbf{x}|\mathbf{z})$.
Therefore, we still have the components of the left eigenvector of
$\mathbf{M}$ under the integral, while the $\mathbf{z}$-dependent
factor can be absorbed into the integration measure. On the other
hand the full inversion can be understood in terms of scattering theory,
where $\Psi(\mathbf{x}|\mathbf{z})$ and $\bar{\Psi}(\mathbf{y|}\mathbf{z})$
play the role of in and out states: the order of momenta gets inverted
after the full scattering of all particles (see e.g. \cite{ruijsenaars}). %
} Correspondingly, given the system has started from a particle configuration
$\mathbf{x^{0}}$, the probability distribution at arbitrary time,
referred to as Green function in this case , is
\begin{equation}
G_{t}(\mathbf{x}|\mathbf{x^{0}})=C_{N}^{-1}P_{st}(\mathbf{x}^{0})\int\Lambda_{N}^{t}(\mathbf{z})\Psi^{0}(\mathbf{\hat{\Pi}x}^{0}|\hat{R}\mathbf{z})\Psi^{0}(\mathbf{x}|\mathbf{z})\mathcal{M}(d\mathbf{z}).\label{eq:Green}
\end{equation}
Our goal is to choose the integration measure and the domain, such
that the relation (\ref{eq:completeness}) holds.

The solution was first proposed in \cite{Schuetz} for the case of
continuous time TASEP. Later this program was completed for a few
models. In all the cases considered to date the integration is performed
along the product of $N$ identical contours $\Gamma_{1}\times\dots\times\Gamma_{N}$
defined by rules of going around singularities of the integrand, and
the measure is $\mathcal{M}(d\mathbf{z})=\bigwedge_{i=1}^{N}dz_{i}/(2\pi\mathrm{i}z_{i})$.

For the TASEP and drop-push models, which corresponds to $q=0,\infty$
of our model, the $S$-matrix possesses special factorization property
$S(u,v)=g(u)/g(v)$, with a rational function $g(u).$ As a result
the function $\Psi^{0}(\mathbf{x}|\mathbf{z})$ has determinantal
form. The integrations in different variables decouple, and the integral
is evaluated explicitly to a determinant of the matrix $N\times N$.
In the simplest cases of the continuous time TASEP and the drop-push model,
whose stationary measure is trivial, this matrix is upper triangular
with the diagonal elements equal to one. A little more complicated case is
the discrete time TASEP with the generalized \cite{Derbyshev Poghosyan  Povolotsky  Priezzhev}
and, in particular, parallel update \cite{Povolotsky Priezzhev}, where the stationary measure is not uniform. Then the integral evaluates to a determinant of a block
diagonal matrix, which yields exactly the inverse stationary measure.
In all these cases the Green function (\ref{eq:Green}) is the determinant
of an $N\times N$ matrix.

The situation is far more complicated when the $S$-matrix can not
be factorized into a product of one variable functions.In this case, the poles of the integrand relate different variables to each other; 
 a fine account
of their contributions is necessary to prove the formulas (\ref{eq:completeness}-\ref{eq:Green}).
This was first implemented for the PASEP by Tracy and Widom \cite{Tracy Widom 1},
who used this result as a starting point of the derivation of the
current distribution. Later, analogous proofs were given for several
other models: the two-sided PushASEP, the asymmetric zero range process
with uniform hopping rates, the asymmetric avalanche process \cite{lee-1}
and the multiparticle hopping asymmetric diffusion model \cite{lee}.

In our case, an explicit substitution of $\Psi^{0}$ to (\ref{eq:completeness-1})
yields two independent sums over permutations $\sigma$ and $\sigma'$.
By changing the summation variable in one of the sums to $\sigma''=\sigma\cdot\sigma'$,
one sum becomes trivial and we arrive at the conjecture
\begin{eqnarray*}
\sum_{\sigma\in S_{N}}\mathrm{sgn}(\sigma)\oint_{\Gamma_{1}^{0,\nu}}\dots\oint_{\Gamma_{N}^{0,\nu}}\prod_{i=1}^{N}\left(\prod_{j>i}\frac{\alpha+\beta z_{\sigma_{i}}+\gamma z_{\sigma_{i}}z_{\sigma_{j}}-z_{\sigma_{j}}}{\alpha+\beta z_{i}+\gamma z_{i}z_{j}-z_{j}}\right)\\
\,\hspace{2cm}\hspace{2cm}\times\frac{z_{\sigma_{i}}^{x_{i}-y_{\sigma_{i}}-1}dz_{i}}{2\pi\mathrm{i}}=\frac{N!C_{N}}{P_{st}(\mathbf{x})}\delta_{\mathbf{x},\mathbf{y}}.
\end{eqnarray*}
Similarly to \cite{Povolotsky Priezzhev} we expect that the contours
must encircle the poles of the integrand at $z_{i}=0,\nu$ leaving
$z_{i}=1$ and $z_{i}=\infty$ outside. The normalization coefficient
$C_{N}$ must be chosen such that the states with all particles being
at different sites are normalized to one. Therefore $$C_{N}=f(1)^{N}/N!.$$

With the use of new variables \eqref{eq:z-u change} our conjecture takes the following form.
\begin{conjecture}
\label{complteteness conjecture} Let $|q|<1$ and $|\nu|<1$. Given
two arbitrary $N-$tuples of integers $\mathbf{x}=(x_{1}\leq\dots\leq x_{N})$
and $\mathbf{y}=(y_{1}\leq\dots\leq y_{N})$, such that
\begin{eqnarray}
y_{1} & = & \dots=y_{n_{1}},\nonumber \\
y_{n_{1}+1} & = & \dots=y_{n_{1}+n_{2}},\label{eq:y_1...y_N}\\
 & \vdots\nonumber \\
y_{n_{1}+n_{2}+\dots+1} & = & \dots=y_{N},\nonumber
\end{eqnarray}
the following identity holds:
\begin{eqnarray}
\sum_{\sigma\in S_{N}}\mathrm{sgn}(\sigma)\oint_{\Gamma_{1}^{0,1}}\dots\oint_{\Gamma_{N}^{0,1}}\prod_{i=1}^{N}\left(\prod_{\begin{array}{c}
j>i\end{array}}\frac{u_{\sigma_{i}}-qu_{\sigma_{j}}}{u_{i}-qu_{j}}\right)\nonumber \\
\,\,\,\,\,\,\times\frac{\left(1-\nu u_{\sigma_{i}}\right)^{x_{i}-y_{\sigma_{i}}-1}}{\left(1-u_{\sigma_{i}}\right)^{x_{i}-y_{\sigma_{i}}+1}}\frac{du_{i}}{2\pi\mathrm{i}}=\delta_{\mathbf{x},\mathbf{y}}(1-q)^{-N}\prod_{\{n_{i}\}}\frac{(q;q)_{n_{i}}}{(\nu,q)_{n_{i}}},\label{eq:conjecture}
\end{eqnarray}
here the integration in $u_{1},\dots,u_{N}$ is performed along the
contours $\Gamma_{1}^{0,1},\dots,\Gamma_{N}^{0,1}$ encircling the
poles of the integrand at $u_{i}=0,1$ and leaving other poles outside.
\end{conjecture}
We have checked the above statement for $N=2$ and leave it a conjecture
for arbitrary $N$, since its proof requires some analytical effort
and is beyond the goals of this paper. The condition $|\nu|<1$ guarantees
that the contours $\Gamma_{i}^{0,1}$ can be chosen being circles
of the radii $1<R<1/\nu$. For larger absolute values of $\nu$ one
has to use the contours deformed accordingly. The condition $|q|<1$
ensures that the poles of the scattering matrix, corresponding to
bound states, do not contribute into the integral. In the case $|q|>1$,
where the bound states are important, they should be taken into account
explicitly. In particular, in the case of the multiparticle hopping
asymmetric diffusion model \cite{lee} the completeness relation required
contours to have a special nested structure.

Given conjecture (\ref{eq:conjecture}), the expression for the Green
function is straightforward
\begin{eqnarray*}
G_{t}(\mathbf{x}|\mathbf{y}) & = & \left(1-q\right)^{N}\prod_{\{n_{i}\}}\frac{(\nu;q)_{n_{i}}}{(q,q)_{n_{i}}},\sum_{\sigma\in S_{N}}\mathrm{sgn}(\sigma)\oint_{\Gamma_{1}^{0,1}}\dots\oint_{\Gamma_{N}^{0,1}}\\
 &  & \prod_{i=1}^{N}\left(\prod_{\begin{array}{c}
j<i\end{array}}\frac{u_{\sigma_{i}}-qu_{\sigma_{j}}}{u_{i}-qu_{j}}\right)\left(\frac{1-\mu u_{i}}{1-\nu u_{i}}\right)^{t}\frac{\left(1-\nu u_{\sigma_{i}}\right)^{x_{i}-y_{\sigma_{i}}-1}}{\left(1-u_{\sigma_{i}}\right)^{x_{i}-y_{\sigma_{i}}+1}}\frac{du_{i}}{2\pi\mathrm{i}}
\end{eqnarray*}
where the numbers $n_{i}$ are defined as in (\ref{eq:y_1...y_N}).

\subsection{Discrete spectrum on the ring\label{sub:Discrete-spectrum-on}}

For the case of the ring the periodic boundary conditions on the eigenfunctions
are imposed
\[
\Psi(x_{1},\dots,x_{N}|\mathbf{z})=\Psi(x_{2},\dots,x_{N},x_{1}+L|\mathbf{z}).
\]
A direct substitution of (\ref{eq:Bethe ansatz},\ref{eq:S-matrix})
gives
\begin{equation}
z_{i}^{L}=\prod_{i\neq j}S(z_{i},z_{j}),\label{eq:Bethe eqs 1}
\end{equation}
where the S-matrix is given in (\ref{eq:S-matrix}). To write down
the periodicity conditions explicitly we again use the variable change
(\ref{eq:z-u change}). In these variables we obtain the following
Bethe ansatz equations:
\[
\left(\frac{1-\nu u_{i}}{1-u_{i}}\right)^{L}=(-1)^{N-1}\prod_{j=1}^{N}\frac{u_{i}-qu_{j}}{u_{j}-qu_{i}},\,\,\, i=1,\dots,N.
\]
The solutions of these equations are to be substituted into the eigenvalues
and eigenvectors (\ref{eq:Lambda(u)}-\ref{eq:Psi(u)}). These equations, the eigenvalues
and the eigenvectors have appeared and been studied before in \cite{Povolotsky Mendes}, where the particular case
our model, discrete time q-ZRP, was discussed. There, however, the parameters $p,q$ and $\nu$ ($-\lambda$ of \cite{Povolotsky Mendes})  were related by  the constraint $\nu=-p/(1-p-q)$, while here they can be considered as three independent quantities.

\subsection{ZRP-ASEP transformation\label{sub:ZRP-ASEP-transformation}}

It was mentioned above that one can construct an ASEP-like process
by replacing a site with $n$ particles by a string of $n$ sites,
occupied by one particle each, plus one empty site ahead. Correspondingly,
the jump of $m$ particles from a site with $n$ particles will be
replaced by a unit right step made by a cluster of $m$ particles
detaching from the right end of an $n$-particle cluster. Technically,
the transformation suggests that the coordinates of particles are
transformed as
\begin{equation}
(x_{1},\dots,x_{N})\to(x_{1}+1,\dots,x_{N}+N),\label{eq:ZRP-ASEP}
\end{equation}
and $N$ extra sites are added to the lattice $L\to L+N$. This is
enough to establish a direct correspondence between finite time realizations
of the processes.

This transformation can also be translated into the language of the
Bethe ansatz. In the case of the ASEP-like dynamics the physically
allowed domain of coordinates is defined by strict inequalities
\begin{equation}
x_{1}<\dots<x_{n}.\label{eq:ASEP domain}
\end{equation}
The equations for $\Psi^{0}(x_{1},\dots,x_{N})$ look as the one for
non-interacting particles within this domain, while the interaction
is set in by imposing the boundary conditions, which express the forbidden
terms like $\Psi^{0}(\dots,x,x,\dots)$ via the allowed ones

\begin{eqnarray}
\Psi^{0}(\dots,x,x,\dots) & = & \alpha\Psi^{0}(\dots,x-1,x,\dots)\nonumber \\
\, & + & \beta\Psi^{0}(\dots,x-1,x+1,\dots)+\gamma\Psi^{0}(\dots,x,x+1,\dots).\label{eq:boundary full-1}
\end{eqnarray}
Again, the eigenvector has the form (\ref{eq:Bethe ansatz}) of
the Bethe ansatz, which yields the same formula (\ref{eq:eigenvalue})
for the eigenvalue. Substituted into the boundary conditions it gives

\[
\frac{A^{ASEP}_{\dots ij\dots}}{A^{ASEP}_{\dots ji\dots}}\equiv-\frac{z_{i}}{z_{j}}S(z_{i},z_{j}),
\]
where $S(z_{i},z_{j})$ is the $S-$matrix (\ref{eq:S-matrix}) obtained
for the ZRP-like dynamics. It is not difficult to see that the eigenfunction $\Psi^{0}(\mathbf{x}|\mathbf{z})$
defined as the Bethe ansatz with the coefficients $A_{\sigma}^{ASEP}$, is nothing but that obtained for ZRP expressed via the new coordinates
of the particles in the ASEP-like system. This obviously leads to
the same, up to a simple coordinate change, Green function as before.

A slight difference appears in the case of the finite ring.  In this case the ZRP-ASEP
transformation gives us one-to-one correspondence of the  realizations of the processes (the sequence of particle jumps), while
the particle configurations on the lattice are equivalent only up to a shift:   when a particle in the
ZRP makes the full rotation around the lattice restoring an original configuration,
 the corresponding  TASEP configuration must be shifted one step back. This correspondence should also be translatable  to the language of
 eigenvectors and eigenvalues,  i.e.  given the set of solutions of one system of the Bethe equations, we should reconstruct the solutions for the other, which, being
 substituted into the eigenvectors and eigenvalues, would provide the equivalence of the time evolutions. This correspondence, however, is hidden in symmetries of the
 Bethe equations and in the properties of the eigenfunctions
 constructed out of their solutions, and making it  explicit   is not a straightforward task.

In particular,
the dimensions of the state spaces, i.e. the total numbers of particle
configurations, are different, being $C_{L+N}^{N}$ for the ASEP-like systems and
 $C_{L+N-1}^{N}$ for the ZRP-like ones. These should be the multiplicities of the solutions of the Bethe equations, as every
 eigenvector corresponds to a unique  solution (provided that all the eigenvectors are linearly independent).
Imposing the periodic boundary conditions $\Psi(x_{1},\dots,x_{N}|\mathbf{z})=\Psi(x_{2},\dots,x_{N},x_{1}+L+N|\mathbf{z})$
on the lattice of the size $L+N,$ we arrive at the system of the
Bethe equations for the TASEP-like system,
\begin{equation}
z_{i}^{L}\prod_{k=1}^{N}z_{k}=\prod_{i\neq j}S(z_{i},z_{j})\,\,\,,i=1,\dots,N\label{eq:Bethe eqs ASEP}
\end{equation}
which differs from (\ref{eq:Bethe eqs 1}) by the factor $\theta=z_{1}\dots z_{N}$
in the l.h.s.. Taking products of all equations in (\ref{eq:Bethe eqs 1})
and in (\ref{eq:Bethe eqs ASEP}), we obtain $\theta^{L}=1$ and $\theta^{L+N}=1$
respectively. Let us use the variable $\theta$ instead of $z_{N}$,
leaving the other $N-1$ variables unchanged. Obviously the equation
for $\theta$ suggests that it takes $L$ and $L+N$ values for the
ZRP and ASEP cases respectively. Given $\theta$ fixed, the other
$N-1$ equations have similar, up to the factor $\theta$ in the l.h.s.,
structure in the two cases. Supposedly the multiplicities of their
solutions are the same. Hence the ratio of the total numbers of solutions
is $L/(L+N)$, which indeed must be the case. Of course, the rigorous proof of this fact  requires
more elaborate arguments, and so do the the proof of the full correspondence.

\section{Conclusion\label{sec:Conclusion}}

To summarize, we have found the three-parametric family of hopping
probabilities for the class of chipping modes with on-site interaction,
which, first, ensure the factorization of the stationary measure on
the infinite lattice and on the ring, and, second, define a Markov matrix
being a transfer matrix of an integrable model solvable by the Bethe
ansatz. Our model contains most of models known to date as particular
limiting cases. We also have given an interpretation of the model
obtained by the ZRP-ASEP transformation in terms of the ASEP-like systems
either with simultaneous jumps of clusters of particles or with long
range single-particle jumps governed by long-range interactions.

We have constructed the Bethe ansatz solution for the infinite lattice
and for the ring. In the former case we have formulated the conjecture
on the form of Green function. If it is proven, the Green function
could be used as a starting point for calculation of the distribution
of the distance traveled by a tagged particle and, potentially, of
the many-particle correlation functions. Nowadays we have a few examples
of solutions of the former problem for different models and several
types of initial conditions. The solution of the latter one is still
an open question.

For the ring we obtained the eigenvectors in terms of solutions of
the system of Bethe equations. This system can be explicitly solved
in very few cases, mostly in the thermodynamic limit. Some results
for infinite time limit can be obtained analogously to \cite{Gwa Spohn,Derrida Lebowitz,Kim}.
The most interesting task is search for correlation functions at finite
time, which would provide us with an information about KPZ-specific
crossover from the results obtained for infinite systems to finite
size behaviour. General correlation functions of the integrable models
on a ring is a long-standing challenging problem having a big history.
Some steps in this direction for stochastic particle models have recently
been done. However, a final solution to this problem has not yet been given.

In the present paper we limited ourselves by considering structural
elements responsible for integrability and did not study the physics
of the model. We expect that the scaling behaviour of the fluctuations
of particle current on the infinite lattice will be similar to that
of other models of KPZ class in the most part of the parameter space.
There are however points in this space where particles either stick
together and move as a single particle or become independent of each
other. It is of interest to study the crossover regimes between KPZ
behaviour and these points to find out how universal they are. Also
the KPZ universality may break down at the point where the particle
current loses its convexity as a function of the particle density.
Whether the particle current in our model have any special points
like that is yet to be studied.

When the article was ready to submission we came to know about a new work
\cite{Borodin-Corwin-Petrov-Sasamoto}, where an elegant Plancherel theory
was developed for the continuous time q-ZRP model.
Among the results, there is a proof of our Conjecture \ref{complteteness conjecture} about completeness of the Bethe
ansatz on the infinite lattice, completed for $\mu=q \nu, \nu\to 0 $  limit of our  model.  The method also exploited the relation
 between  forward and backward dynamics, similar to relation (\ref{symmetry}) between the Markov matrix
and its transpose. It is of interest to extend the technique of that paper to the case of three parameters.

In another article \cite{Korhonen Lee}, appeared right after our article was submitted to the journal, the formula of
the Green function for  continuous time q-ZRP model was also proved, which was then used to obtain
the integral representation for the distribution of  the left-most particle's position.

\ack{}{}

This work is supported by the RFBR grant 12-01-00242-a and the grant
of the Heisenberg-Landau program. The author is indebted to V.B. Priezzhev
for stimulating discussions of interacting particle systems, to P.N.
Pyatov for his advice about quantum binomial formulas and diamond
lemma and to V.P. Spriridonov for his comments on q-functions.

\appendix

\section{Proof of Theorem 1.\label{sec: app1}}

We need to prove that given the generators $A$ and $B$ satisfying
quadratic relation (\ref{eq: relations}), the expansion (\ref{eq:skew binomial})
holds with the expansion coefficients  (\ref{eq:distrib}),
\[
\varphi(m|n)=\mu^{m}\frac{(\nu/\mu;q)_{m}(\mu;q)_{n-m}}{(\nu;q)_{n}}\frac{(q;q)_{n}}{(q;q)_{m}(q;q)_{n-m}},
\]
where the parameters $q,\nu$ and $\mu$ parameterize $\alpha,\beta,\gamma$
and $p$ in accordance with (\ref{eq:a,b,c (parametrized)},\ref{eq:mu}).
The proof is inductive. The statement obviously holds for $n=1$.
Indeed, in this case $\varphi(0|1)=(1-\mu)/(1-\nu)=1-p$ and $\varphi(1|1)=(\nu-\mu)/(1-\nu)=p.$
We suppose that it is true for $n-1$ and prove it for $n$. Let us
rewrite $(pA+(1-p)B)^{n}$ in form $(pA+(1-p)B)^{n-1}(pA+(1-p)B)$
and apply the expansion (\ref{eq:distrib}) to the first term
\begin{equation}
(pA+(1-p)B)^{n}=\sum_{m=0}^{n}\varphi(m|n-1)\left[pA^{m}B^{n-m-1}A+(1-p)A^{m}B^{n-m}\right].\label{eq:(pA+(1-p)B)^{n+1}}
\end{equation}
To find $\varphi(n|m)$ we want the summands being normally ordered
words. What violates the normal order is the factor $B^{n-m-1}A$.
Thus we need to find expansion coefficients for the words of this
kind.

Let us suppose that for any $l\geq1,$ expansion

\begin{equation}
B^{l-1}A=\sum_{k=0}^{l}a_{k}^{l}A^{l-k}B{}^{k},\label{eq:B^{l-1}A}
\end{equation}
holds, where $a_{k}^{l}$ are the coefficients to be found. In particular,
for $l=1,2$ we have

\begin{eqnarray}
a_{0}^{1} & = & 1,a_{1}^{1}=0,\label{eq: a^1}\\
a_{0}^{2} & = & \alpha,a_{1}^{2}=\beta,a_{2}^{2}=\gamma.\label{eq:a^2}
\end{eqnarray}
 With such defined $a_{k}^{l}$, we can rewrite (\ref{eq:B^{l-1}A})
as
\begin{eqnarray*}
(pA+(1-p)B)^{n} & = & (1-p)\sum_{l=0}^{n-1}\varphi(l|n-1)A^{l}B^{n-l}\\
 & +p & \sum_{l=0}^{n}\left[\sum_{m=0}^{l}\varphi(m|n-1)a_{n-l}^{n-m}\right]A^{l}B^{n-l},
\end{eqnarray*}
which suggests
\begin{eqnarray}
\varphi(l|n) & = & p\sum_{m=0}^{l}\varphi(m|n-1)a_{n-l}^{n-m}+(1-p)\varphi(l|n-1),\,\,\, l<n\label{eq: phi(l|n)}\\
\varphi(n|n) & = & p\sum_{m=0}^{n-1}\varphi(m|n-1)a_{0}^{n-m}\label{eq: phi(n|n)}
\end{eqnarray}
 Therefore, before proceeding with finding, $\varphi(l|n)$ we first
need to find coefficients $a_{m}^{n}$. To this end we note that the
expansion (\ref{eq:B^{l-1}A}) is a consequence of a successive application
of the single relation (\ref{eq: relations}). Therefore, the coefficients
of interest satisfy  a set of constraints. To write down the constraints
let us represent $B^{l}A$ as $B^{l-1}(BA)$ and apply the relation
(\ref{eq: relations}) to the second factor.
\begin{eqnarray*}
B^{l}A & = & B^{l-1}(\alpha A^{2}+\beta AB+\gamma B^{2})\\
 & = & \sum_{k=0}^{l}a_{k}^{l}(\alpha A^{l-k}B{}^{k}A+\beta A^{l-k}B{}^{k+1})+\gamma B^{l+1}\\
 & = & \alpha\sum_{k=0}^{l}\sum_{j=0}^{k+1}a_{k}^{l}a_{j}^{k+1}A^{l-j+1}B{}^{j}+\beta\sum_{k=0}^{l}a_{k}^{l}A^{l-k}B^{k+1}+\gamma B^{l+1}\\
 & = & \alpha\sum_{j=0}^{l+1}A^{l-j+1}B{}^{j}\left(\sum_{k=j-1}^{l}a_{k}^{l}a_{j}^{k+1}\right)+\beta\sum_{j=1}^{l+1}a_{j-1}^{l}A^{l-j+1}B^{j}+\gamma B^{l+1}
\end{eqnarray*}
Here we used the expansion (\ref{eq:B^{l-1}A}) applied to $B^{l-1}A$
in the second line and applied to $B{}^{k}A$ in the third line and
then exchanged the summation order in the last line. Doing this we
adopted a convention $a_{j}^{l}=0$ for $j<0$ and $l>0$. Collecting
factors coming with $A^{l+1-j}B^{j}$ for $j=0,\dots,l+1$ and bringing
the term containing $a_{j}^{l+1}$ to the l.h.s., we have
\begin{equation}
a_{j}^{l+1}=(1-\alpha a_{l}^{l})^{-1}\left(\alpha\sum_{k=j-1}^{l-1}a_{k}^{l}a_{j}^{k+1}+\beta a_{j-1}^{l}+\delta_{j,l+1}\gamma\right).\label{eq:recursion}
\end{equation}
One can see that the set of relations has a triangular structure,
i.e. $a_{j}^{l}$ can be expressed in terms of $a_{i}^{k}$ with $k\leq l$
and $k-i\leq l-j$ only. Therefore, we first can find $a_{i}^{l}$
with $l=i,$ then with $l=i+1$ e.t.c..

$a_{l}^{l}:$ The sum in the r.h.s. of (\ref{eq:recursion}) is empty
and the relations are reduced to a simple recursion
\[
a_{l+1}^{l+1}=\frac{\gamma+\beta a_{l}^{l}}{1-\alpha a_{l}^{l}}
\]
with the initial condition $a_{1}^{1}=0$. This is a Riccati difference
equation, which can be linearized by the variable change $a_{l}^{l}=1+\eta/c_{l},$
where $\eta$ is chosen such that the equation for $c(l)$ is linear.
At this point it is more convenient to use the parameters $\nu$ and
$q$ instead of $\alpha,\beta,\gamma$. In terms of these parameters
we can choose either $\eta=1$ or $\eta=1/\nu$. Then we obtain two
relations, which go into one another under the change $q\to1/q$.
Therefore, without loss of generality we choose $\eta=1/\nu$, which
yields relation
\[
c_{l+1}=qc_{l}-\frac{\nu(1-q)}{1-\nu}
\]
subject to initial condition $c_{1}=-\nu.$ This recursion can be
solved to
\[
c_{l}=\frac{\nu(1-\nu q^{l-1})}{\nu-1},
\]
which, after going back to $a_{l}^{l}$ yields the result
\[
a_{l}^{l}=\frac{1-q^{l-1}}{1-\nu q^{l-1}}.
\]

$a_{l}^{l+1}:$ The sum in the r.h.s. of (\ref{eq:recursion}) consists
of a single term and we obtain
\begin{eqnarray*}
a_{l}^{l+1} & = & a_{l-1}^{l}\frac{\beta+\alpha a_{l}^{l}}{1-\alpha a_{l}^{l}}\\
 & =a_{0}^{1} & \prod_{k=1}^{l}\frac{\beta+\alpha a_{k}^{k}}{1-\alpha a_{k}^{k}}\\
 & = & \frac{q^{k}(q-\nu)(1-\nu)}{\left(1-\nu q^{k-1}\right)\left(1-\nu q^{k}\right)}.
\end{eqnarray*}
We can make a few more steps in this way. At every step we obtain
a linear recurrent relation that can be iterated and solved subject
to the initial condition $a_{-1}^{l}=0.$ However, the calculations
quickly become too involved. Fortunately, we can use the results of
the first few steps to make an educated guess about general structure
of $a_{k}^{l}$, which then can be proved by induction.
\begin{lem}
For $l>0$ and $k<l$ we have
\begin{equation}
a_{k}^{l}=(1-\nu)(q-\nu)\nu^{l-k-1}q^{k-1}\frac{(q;q)_{l-1}(\nu;q)_{k-1}}{(q;q)_{k}(\nu;q)_{l}}\,\,\,\mathrm{for}\,\,\, n>k\label{eq: a^l_k}
\end{equation}
 and
\begin{equation}
a_{l}^{l}=\frac{1-q^{l-1}}{1-\nu q^{l-1}}.\label{eq:a_l^l}
\end{equation}
\end{lem}
\begin{proof}
It follows from (\ref{eq: a^1},\ref{eq:a_l^l}) that the formulas
(\ref{eq: a^l_k},\ref{eq:a_l^l}) hold in the cases $l=1,2$. Suppose
they also hold for $a_{j}^{l}$ for $l\leq n$ and $j\leq l$. To
prove them for $l=n+1$ we substitute (\ref{eq: a^l_k},\ref{eq:a_l^l})
into the r.h.s. of (\ref{eq:recursion}), which, after tedious but
elementary algebra, yield the desired result. The main ingredient
of the calculation is evaluation of the sum, which turns out to have
a telescopic structure.
\end{proof}
Now we are in a position to complete the proof of the form of $\varphi(m|n).$
In fact, the principle a proof by induction simply requires that we
substitute (\ref{eq:distrib}) into the r.h.s. of (\ref{eq: phi(l|n)},\ref{eq: phi(n|n)})
and show that the l.h.s. also complies with this formula. This indeed
what will finally be done. However, we first show some steps leading
to a final formula and simplifying the problem to expressions, from
which it can be guessed.

Let us use the ansatz (\ref{eq: phi(n|m)}),
\[
\varphi(m|n)=\frac{v(m)w(n-m)}{f(n)},\,\,\,\mathrm{where}\,\,\, f(n)=\sum_{i=0}^{n}v(i)w(n-i),
\]
 which was the ansatz for the hopping probabilities that ensured the
factorized structure of the stationary state. We recall that $v(m)$
and $w(m)$ are arbitrary positive valued function of $m,$ for which
we have fixed $v(0)=w(0)=1,$ while another free parameter $v(1)$
can yet be fixed without loss of generality. Note that it is not obvious
from eqs. (\ref{eq: phi(l|n)},\ref{eq: phi(n|n)}) that $\varphi(m|n)$
must have this structure. However, once we have found the solution
of this form, which complies with the initial conditions, an induction
will provide its uniqueness. Let us write the equation (\ref{eq: phi(l|n)})
for the case $l=0,$ remembering that we assigned $v(0)=1.$
\begin{eqnarray}
\frac{w(n)}{f(n)} & = & \frac{w(n-1)}{f(n-1)}(pa_{n}^{n}+(1-p))\nonumber \\
 & = & \frac{w(0)}{f(0)}\prod_{k=1}^{n}(pa_{k}^{k}+(1-p))\label{eq:l=00003D0}\\
 & = & \frac{(\mu;q)_{n}}{(\nu;q)_{n}},\nonumber
\end{eqnarray}
where in the second line we iterated the recurrence from the first
line and substituted explicitly $a_{k}^{k}$ and $p$ in the third
line remembering that $w(0)=1$ and $f(0)=1$.

Next, we write the same relation for $l=1\text{:}$
\begin{eqnarray*}
\frac{v(1)w(n-1)}{f(n)} & = & p\frac{w(n-1)}{f(n-1)}a_{n-1}^{n}+\frac{v(1)w(n-2)}{f(n-1)}\left(pa_{n-1}^{n-1}+(1-p)\right)\\
 & = & p\frac{w(n-1)}{f(n-1)}a_{n-1}^{n}+\frac{v(1)w(n-1)f(n-2)}{f(n-1)^{2}},
\end{eqnarray*}
where going from the first line to the second one we used the first
line of (\ref{eq:l=00003D0}) to rewrite $\left(pa_{n-1}^{n-1}+(1-p)\right)$.
In this way we obtain the recurrent relations for the ratio $f(n-1)/f(n),$
\begin{eqnarray*}
\frac{f(n-1)}{f(n)} & = & \frac{p}{v(1)}a_{n-1}^{n}+\frac{f(n-2)}{f(n-1)}\\
 & =\frac{p}{v(1)} & \left(\sum_{k=2}^{n}a_{k-1}^{k}+1\right),
\end{eqnarray*}
where we used the fact that $f(0)=1$ and $f(1)=v(1)/p$ and $n\geq2$.
Then, we have
\[
f(n)=\left(\frac{v(1)}{p}\right)^{n}\prod_{l=2}^{n}\left(\sum_{k=2}^{n}a_{k-1}^{k}+1\right)^{-1}.
\]
 Substituting $a_{k-1}^{k}$
\[
a_{k-1}^{k}=\frac{(1-\nu)(q-\nu)q^{k-2}}{(1-\nu q^{k-1})(1-\nu q^{k-2})}
\]
 and evaluating telescoping sum $\sum_{k=1}^{n}a_{k-1}^{k}=(1-\nu)(1-q^{n})/((1-q)(1-\nu q^{n-1}))$
we obtain
\[
f(n)=\left(\frac{v(1)(1-q)}{p(1-\nu)}\right)^{n}\frac{(\nu;q)_{n}}{(q;q)_{n}}.
\]
 Now we recall that fixing $v(1)$ does not affect the form of $\varphi(m|n)$.
For convenience we fix it such that $f(n)$  has no exponential
part,
\[
\frac{v(1)(1-q)}{p(1-\nu)}=1,
\]
which yields the formula (\ref{eq:f(n)}),
\[
f(n)=\frac{(\nu;q)_{n}}{(q;q)_{n}},
\]
 for single site weight announced in section \ref{sec:Model-and-results}.
Then, comparing this result with (\ref{eq:l=00003D0}) we obtain
\[
w(n)=\frac{(\mu;q)_{n}}{(q;q)_{n}}
\]
 given in (\ref{eq:v(k), w(k)}).

With $w(n)$ and $f(n)$ in hands, eqs. (\ref{eq: phi(l|n)},\ref{eq: phi(n|n)})
become infinite hierarchy of equations for still unknown function
$v(m)$,
\begin{eqnarray}
v(l) & = & (\mu-\nu)\nu^{l-1}\frac{(\nu;q)_{n-l}}{(\mu;q)_{n-l}}\frac{1-\nu q^{n-1}}{1-q^{l}}\sum_{m=0}^{l-1}v(m)\frac{(\mu;q)_{n-m-1}}{\nu^{m}(\nu;q)_{n-m}},\label{eq:rec v(l)}
\end{eqnarray}
where $l\leq n$ and $n=2,3,\dots.$ For every $n$ these are recurrent
relations to be solved with the initial condition $v(1)=(\mu-\nu)/(1-q)$.
At the first glance, the problem seems being overdetermined, as we
have the infinite set of equations determining $v(l)$ at every $l.$
However, it is straightforward to check directly that the result does
not depend on $n$ and all the equations are solved by a single function
announced in (\ref{eq:v(k), w(k)}),
\[
v(m)=\mu^{m}\frac{(\nu/\mu;q)_{m}}{(q;q)_{m}}.
\]
 The proof is inductive. Obviously the initial conditions are satisfied.
Let us fix $n$ and suppose that this formula is valid for any $m<l\leq n$.
We introduce an auxiliary function
\[
s_{n,k}=\sum_{m=0}^{l}v(m)\frac{(\mu;q)_{n-m-1}}{\nu^{m}(\nu;q)_{n-m}},
\]
 which is the partial sum from the r.h.s. of (\ref{eq:rec v(l)}).
It can be directly summed to
\begin{equation}
s_{n,k}=\frac{1}{1-\nu q^{n-1}}\frac{\mu}{\mu-\nu}\left(\frac{\mu}{\nu}\right)^{k}\frac{(\mu;q)_{n-k-1}(\nu/\mu;q)_{k+1}}{(q;q)_{k}(\nu;q)_{n-k-1}},\label{eq:s_{n,k}}
\end{equation}
which is proved by another induction in $k$. Specifically, the formula
holds for $s_{n,0}$ and identity
\[
s_{n,k}=s_{n,k-1}+v(k)\frac{(\mu;q)_{n-k-1}}{\nu^{k}(\nu;q)_{n-k}}
\]
 can be checked directly. It follows from (\ref{eq:rec v(l)}) that
\[
v(l)=(\mu-\nu)\nu^{l-1}\frac{(\nu;q)_{n-l}}{(\mu;q)_{n-l}}\frac{1-\nu q^{n-1}}{1-q^{l}}s_{n,l-1},
\]
and substituting (\ref{eq:s_{n,k}}) we arrive at the expression for
$v(l)$ from (\ref{eq:v(k), w(k)}), which is independent of $n$.

Finally, we have shown that the recursion relations (\ref{eq: phi(l|n)},\ref{eq: phi(n|n)})
hold for $v(l)$, $w(l)$ and $f(n)$ found. Note that in the above
proof we did not use the fact that $f(n)$ is the convolution of $v(n)$
and $w(n)$, except for $f(1).$ This fact, however, being nothing
but the normalization condition, is related to the probabilistic nature
of the problem. Indeed the corresponding three parametric generalization
of the binomial theorem,
\[
\sum_{m=0}^{n}\mu^{m}(\nu/\mu;q)_{m}(\mu;q)_{n-m}\left[\begin{array}{c}
n\\
m
\end{array}\right]=(\nu;q)_{n},
\]
 proved in the theory of basic hypergeometric series can be found
in \cite{Gasper Rahman}.

\end{document}